\documentclass[12pt]{article}

\usepackage{amsmath}
\usepackage{amssymb}
\usepackage{amsfonts}
\usepackage{latexsym}
\usepackage{color}

\catcode `\@=11 \@addtoreset{equation}{section}

\catcode `\@=12

\newcommand{\be}{\begin{equation}}
\newcommand{\en}{\end{equation}}
\newcommand{\bea}{\begin{eqnarray}}
\newcommand{\ena}{\end{eqnarray}}
\newcommand{\beano}{\begin{eqnarray*}}
\newcommand{\enano}{\end{eqnarray*}}
\newcommand{\bee}{\begin{enumerate}}
\newcommand{\ene}{\end{enumerate}}

\newcommand{\mc}{\mathcal}

\newcommand{\D}{{\mc D}}

\newcommand{\E}{{\cal E}}
\newcommand{\F}{{\cal F}}
\newcommand{\G}{{\cal G}}

\newcommand{\Lc}{{\cal L}}

\newcommand{\1}{1 \!\! 1}

\newcommand{\Hil}{\mc H}

\newtheorem{thm}{Theorem}

\newtheorem{prop}[thm]{Proposition}
\newtheorem{defn}[thm]{Definition}

\newenvironment{proof}{\noindent {\bf Proof --}}{\hfill$\square$ \vspace{3mm}\endtrivlist}

\catcode `\@=11 \@addtoreset{equation}{section}
\catcode `\@=12

\begin{document}

\thispagestyle{empty}

\vspace*{2cm}

\begin{center}
{\Large \bf Pseudo-bosons and Riesz bi-coherent states\footnote{This paper is dedicated to the memory of Gerard Emch}}   \vspace{2cm}\\

{\large F. Bagarello}\\
  Dipartimento di Energia, Ingegneria dell'Informazione e Modelli Matematici,\\
Facolt\`a di Ingegneria, Universit\`a di Palermo,\\ I-90128  Palermo, Italy\\
and I.N.F.N., Sezione di Torino\\
e-mail: fabio.bagarello@unipa.it\\
home page: www.unipa.it/fabio.bagarello

\end{center}

\vspace*{2cm}

\begin{abstract}
\noindent After a brief review on $\D$-pseudo-bosons we introduce what we call {\em Riesz bi-coherent states}, which are pairs of states sharing with ordinary coherent states most of their features. In particular, they produce a resolution of the identity and they are eigenstates of two different annihilation operators which obey pseudo-bosonic commutation rules.

\end{abstract}

\vspace{2cm}


\vfill


\newpage

\section{Introduction}

In a series of papers the notion of $\D$-pseudo bosons ($\D$-PBs) has been introduced and studied in many details. We refer to \cite{baginbagbook} for a recent review on this subject, and for more references. In particular, we have analyzed the functional structure arising from two operators $a$ and $b$, acting on a Hilbert space $\Hil$ and satisfying, in a suitable sense, the pseudo-bosonic commutation rule $[a,b]=\1$. Here $\1$ is the identity operator. We have shown how two biorthogonal families of eigenvectors of two non self-adjoint operators can be easily constructed, having real eigenvalues, and we have discussed how and when these operators are similar to a single self-adjoint number operator, and which kind of intertwining relations can be deduced. We have also seen that this settings is strongly related to physics, and in particular to $PT$-quantum mechanics, \cite{ben,mosta}, since many models originally introduced in that context can be written in terms of $\D$-PBs.

In connection with $\D$-PBs, the notion of bicoherent states, originally introduced in \cite{tri}, has been considered in some of its aspects, see \cite{bagpb1,abg2015}. Since $a$ and $b$ are unbounded, several mathematical subtle points need to be considered when dealing with these states, as it is clear from the treatment in \cite{abg2015}. However, it is possible, and instructive, to consider a simpler situation, and this is exactly what we will do in this paper: more explicitly,  we will adapt the notion of Riesz bases to coherent states, introducing what we can call {\em Riesz bicoherent states} (RBCS), and we will study some of their features.

This article is organized as follows: in the next section, to keep the paper self-contained, we review few facts on $\D$-PBs. In Section III we introduce our RBCS and analyze their properties, while our conclusions and plans for the future are discussed in Section IV.

\section{Few facts on  $\D$-PBs}

We briefly review here few facts and definitions on $\D$-PBs. More details can be found in  \cite{baginbagbook}.

Let $\Hil$ be a given Hilbert space with scalar product $\left<.,.\right>$ and related norm $\|.\|$. Let further $a$ and $b$ be two operators
on $\Hil$, with domains $D(a)$ and $D(b)$ respectively, $a^\dagger$ and $b^\dagger$ their adjoint, and let $\D$ be a dense subspace of $\Hil$
such that $a^\sharp\D\subseteq\D$ and $b^\sharp\D\subseteq\D$, where $x^\sharp$ is $x$ or $x^\dagger$. Of course, $\D\subseteq D(a^\sharp)$
and $\D\subseteq D(b^\sharp)$.

\begin{defn}\label{def21}
The operators $(a,b)$ are $\D$-pseudo bosonic ($\D$-pb) if, for all $f\in\D$, we have
\be
a\,b\,f-b\,a\,f=f.
\label{A1}\en
\end{defn}

\vspace{2mm}

Our  working assumptions are the following:

\vspace{2mm}

{\bf Assumption $\D$-pb 1.--}  there exists a non-zero $\varphi_{ 0}\in\D$ such that $a\,\varphi_{ 0}=0$.

\vspace{1mm}

{\bf Assumption $\D$-pb 2.--}  there exists a non-zero $\Psi_{ 0}\in\D$ such that $b^\dagger\,\Psi_{ 0}=0$.

\vspace{2mm}

Then, if $(a,b)$ satisfy Definition \ref{def21}, it is obvious that $\varphi_0\in D^\infty(b):=\cap_{k\geq0}D(b^k)$ and that $\Psi_0\in D^\infty(a^\dagger)$, so
that the vectors \be \varphi_n:=\frac{1}{\sqrt{n!}}\,b^n\varphi_0,\qquad \Psi_n:=\frac{1}{\sqrt{n!}}\,{a^\dagger}^n\Psi_0, \label{A2}\en
$n\geq0$, can be defined and they all belong to $\D$ and, as a consequence, to the domains of $a^\sharp$, $b^\sharp$ and $N^\sharp$, where $N=ba$. We further introduce $\F_\Psi=\{\Psi_{ n}, \,n\geq0\}$ and
$\F_\varphi=\{\varphi_{ n}, \,n\geq0\}$.

It is now simple to deduce the following lowering and raising relations:
\be
\left\{
    \begin{array}{ll}
b\,\varphi_n=\sqrt{n+1}\varphi_{n+1}, \qquad\qquad\quad\,\, n\geq 0,\\
a\,\varphi_0=0,\quad a\varphi_n=\sqrt{n}\,\varphi_{n-1}, \qquad\,\, n\geq 1,\\
a^\dagger\Psi_n=\sqrt{n+1}\Psi_{n+1}, \qquad\qquad\quad\, n\geq 0,\\
b^\dagger\Psi_0=0,\quad b^\dagger\Psi_n=\sqrt{n}\,\Psi_{n-1}, \qquad n\geq 1,\\
       \end{array}
        \right.
\label{A3}\en as well as the eigenvalue equations $N\varphi_n=n\varphi_n$ and  $N^\dagger\Psi_n=n\Psi_n$, $n\geq0$. In particular, as a consequence
of these two last equations,  choosing the normalization of $\varphi_0$ and $\Psi_0$ in such a way $\left<\varphi_0,\Psi_0\right>=1$, we deduce that
\be \left<\varphi_n,\Psi_m\right>=\delta_{n,m}, \label{A4}\en
 for all $n, m\geq0$. Hence $\F_\Psi$ and $\F_\varphi$ are biorthogonal. Our third assumption is the following:

\vspace{2mm}

{\bf Assumption $\D$-pb 3.--}  $\F_\varphi$ is a basis for $\Hil$.

\vspace{1mm}

This is equivalent to requiring that $\F_\Psi$ is a basis for $\Hil$ as well, \cite{chri}. However, several  physical models suggest to adopt the following weaker version of this assumption, \cite{baginbagbook}:

\vspace{2mm}

{\bf Assumption $\D$-pbw 3.--}  For some subspace $\G$ dense in $\Hil$, $\F_\varphi$ and $\F_\Psi$ are $\G$-quasi bases.

\vspace{2mm}
This means that, for all $f$ and $g$ in $\G$,
\be
\left<f,g\right>=\sum_{n\geq0}\left<f,\varphi_n\right>\left<\Psi_n,g\right>=\sum_{n\geq0}\left<f,\Psi_n\right>\left<\varphi_n,g\right>,
\label{A4b}
\en
which can be seen as a weak form of the resolution of the identity, restricted to $\D$.
To refine further the structure, in \cite{baginbagbook} we have assumed that a self-adjoint, invertible, operator $\Theta$, which leaves, together with $\Theta^{-1}$, $\D$ invariant, exists: $\Theta\D\subseteq\D$, $\Theta^{-1}\D\subseteq\D$. Then we say that $(a,b^\dagger)$ are $\Theta-$conjugate if $af=\Theta^{-1}b^\dagger\,\Theta\,f$, for all $f\in\D$. One can prove that, if $\F_\varphi$ and $\F_\Psi$ are $\D$-quasi bases for $\Hil$, then the operators $(a,b^\dagger)$ are $\Theta-$conjugate if and only if $\Psi_n=\Theta\varphi_n$, for all $n\geq0$. Moreover, if $(a,b^\dagger)$ are $\Theta-$conjugate, then $\left<f,\Theta f\right>>0$ for all non zero $f\in \D$.

In the rest of the paper, rather than using Assumption $\D$-pbw 3., we will consider the following stronger version:

\vspace{2mm}

{\bf Assumption $\D$-pbs 3.--}  $\F_\varphi$ is a Riesz basis for $\Hil$.

\vspace{1mm}

This implies that a bounded operator $S$, with bounded inverse $S^{-1}$, exists in $\Hil$, together with an orthonormal basis $\F_e=\{e_n,\,n\geq0\}$, such that $\varphi_n=Se_n$, for all $n\geq0$. Then, because of the uniqueness of the  basis biorthogonal to $\F_\varphi$, it is clear that $\F_\Psi$ is also a Riesz basis for $\Hil$, and that $\Psi_n=(S^{-1})^\dagger e_n$. Hence, putting $\Theta:=(S^\dagger S)^{-1}$, we deduce that $\Theta$ is also bounded, with bounded inverse, is self-adjoint, positive, and that $\Psi_n=\Theta \varphi_n$, for all $n\geq0$. $\Theta$ and $\Theta^{-1}$ can be both written as a series of rank-one operators. In fact, adopting the Dirac bra-ket notation, we have
$$
\Theta=\sum_{n=0}^\infty |\Psi_n\left>\right<\Psi_n|,\qquad \Theta^{-1}=\sum_{n=0}^\infty |\varphi_n\left>\right<\varphi_n|.
$$
Of course both $|\Psi_n\left>\right<\Psi_n|$ and $|\varphi_n\left>\right<\varphi_n|$ are not projection operators\footnote{Here $\left(|f\left>\right<f|\right)g=\left<f,g\right>f$, for all $f,g\in\Hil$.} since, in general the norms of $\Psi_n$ and $\varphi_n$ are not equal to one.

Notice now that, calling $\Lc_\varphi$ and $\Lc_\Psi$ the linear span of $\F_\varphi$ and $\F_\Psi$ respectively, both sets are contained in $\D$ and dense in $\Hil$. Moreover, $\Theta:\Lc_\varphi\rightarrow\Lc_\Psi$, so that it is quite natural to imagine that $\Theta$ also maps $\D$ into itself. This is, in fact, ensured if both $S^\sharp$ and $(S^{-1})^\sharp$ map $\D$ into $\D$, condition which is satisfied in several explicit models, and for this reason will always be assumed here. Hence, both $\Theta$ and $\Theta^{-1} $map $\D$ into itself. Of course, this assumption also guarantees that $e_n\in\D$, for all $n$.

The lowering and raising conditions in (\ref{A3}) for $\varphi_n$ can be rewritten in terms of $e_n$ as follows:
\be
S^{-1}aSe_n=\sqrt{n}\,e_{n-1},\qquad S^{-1}bSe_n=\sqrt{n+1}\,e_{n+1},
\label{a5}\en
for all $n\geq0$. Notice that we are putting $e_{-1}\equiv0$. It is now possible to check that
$$
S^\dagger b^\dagger S^{-1}f=S^{-1} a Sf, \qquad S^\dagger a^\dagger S^{-1}f=S^{-1} b Sf,
$$
for all $f\in\D$. Also, the first equation in (\ref{a5}) suggests to define an operator $c$ acting on $\D$ as follows: $cf=S^{-1}aSf$. Of course, if we take $f=e_n$, we recover (\ref{a5}). Moreover, simple computations show that $c^\dagger$ satisfies the equality $c^\dagger f=S^{-1}bS f$, $f\in\D$, which again, taking $f=e_n$, produces the second equality in (\ref{a5}). These operators satisfy the canonical commutation relation (CCR) on $\D$: $[c,c^\dagger]f=f$, $\forall f\in\D$.

We end this section by noticing that, since each pair of biorthogonal Riesz bases are also $\D$-quasi bases, Proposition 3.2.3 of \cite{baginbagbook} implies that $(a,b^\dagger)$ are $\Theta$-conjugate: $af=\Theta^{-1}b^\dagger \Theta f$, $\forall f\in\D$, and that $\Theta$ is positive, as we have already noticed because of its explicit form.

\section{Riesz bicoherent states}

In \cite{bagpb1,abg2015} we have considered the notion of bicoherent states, and we have deduced some of their properties. Here we discuss a somehow stronger version of these states, which we call Riesz bicoherent states (RBCS).

We start recalling that, calling $W(z)=e^{zc^\dagger-\overline{z}\,c}$, a {\em standard} coherent state is the vectors
\be
\Phi(z)=W(z)e_0=e^{-|z|^2/2}\sum_{k=0}^\infty \frac{z^k}{\sqrt{k!}}\,e_k.
\label{30}\en
Here $c$ and $c^\dagger$ are operators satisfying the CCR, and $\F_e$ is the orthonormal basis related to these operators as shown in Section II. The vector $\Phi(z)$ is well defined, and normalized, for all $z\in\Bbb C$. This is just a consequence of the fact that $W(z)$ is unitary, or, alternatively, of the fact that $\left<e_k,e_l\right>=\delta_{k,l}$. Moreover,
$$
c\,\Phi(z)=z\Phi(z),\qquad\mbox{and}\qquad \frac{1}{\pi}\int_{\Bbb C}d^2z|\Phi(z)\left>\right<\Phi(z)|=\1.
$$
It is also well known that $\Phi(z)$ saturates the Heisenberg uncertainty relation, which will not be discussed in this paper.

What is interesting to us here is whether the family of vectors $\{\Phi(z),\,z\in\Bbb C\}$ can be somehow generalized in order to recover similar properties, and if this generalization is related to the pseudo-bosonic operators $a$ and $b$ introduced in the previous section. For that, let us introduce the following operators:
\be
U(z)=e^{zb-\overline{z}\,a},\qquad V(z)=e^{za^\dagger-\overline{z}\,b^\dagger}.
\label{31}\en
Of course, if $a=b^\dagger$, then $U(z)=V(z)$ and the operator is unitary and essentially coincide with $W(z)$, with $a\equiv c$. However, the case of interest here is when $a\neq b^\dagger$. In \cite{bagpb1,abg2015} we have introduced the vectors
\be
\varphi(z)=U(z)\varphi_0,\qquad \Psi(z)=V(z)\,\Psi_0.
\label{32}\en
They surely exist if $z=0$. We will see that, in the present working conditions, they are well defined in $\Hil$ for all $z\in\Bbb{C}$. A way to prove this result is to use the Baker-Campbell-Hausdorff formula which produces the identities
$$
U(z)=e^{-|z|^2/2}\,e^{z\,b}\,e^{-\overline{z}\,a},\qquad V(z)=e^{-|z|^2/2}\,e^{z\,a^\dagger}\,e^{-\overline{z}\,b^\dagger}.
$$
Then,
\be
\varphi(z)=e^{-|z|^2/2}\,\sum_{n=0}^\infty\,\frac{z^n}{\sqrt{n!}}\,\varphi_n, \qquad
\Psi(z)=e^{-|z|^2/2}\,\sum_{n=0}^\infty\,\frac{z^n}{\sqrt{n!}}\,\Psi_n.
\label{33}\en
These clearly extend formula (\ref{30}) for $\Phi(z)$.
Now, \cite{bagpb1}, since $\|\varphi_n\|=\|Se_n\|\leq\|S\|$ and $\|\Psi_n\|=\|(S^{-1})^\dagger e_n\|\leq\|S^{-1}\|$, the two series converge for all $z\in\Bbb C$. Hence both $\varphi(z)$ and $\Psi(z)$ are defined everywhere in the complex plane. Incidentally we observe that this is different from what happens in \cite{abg2015}, where $\F_\varphi$ and $\F_\Psi$ are not assumed to be Riesz bases, and some estimate must be assumed on $\|\varphi_n\|$ and $\|\Psi_n\|$. Also in view of possible applications, and in particular of the relation with Definition \ref{def3} below, it is interesting to show how to deduce the same result (i.e. $\varphi(z)$ and $\Psi(z)$ are defined everywhere) using a different strategy, assuming that $a$, $b$ and $c$ are related as in Section II.

The key of this strategy is the following

\begin{prop}
With the above definitions the following equalities hold:
\be
U(z)f=SW(z)S^{-1}f, \qquad \mbox{and}\qquad V(z)f=(S^{-1})^\dagger W(z)S^\dagger f
\label{34}
\en
for all $f\in\D$.
\end{prop}

\begin{proof}
We prove here the first equality. The second can be proved in a similar way.

First of all we can prove, by induction, that, for all $f\in\D$ and for all $k=0,1,2,3,\ldots$,
\be
S\left(zc^\dagger-\overline{z}\,c\right)^kS^{-1}f=\left(zb-\overline{z}\,a\right)^kf.
\label{35}\en
This equality is evident for $k=0$. This equality for $k=1$ follows from the equations $cf=S^{-1}aSf$ and $c^\dagger f=S^{-1}bS f$, $f\in\D$. Now, assuming that this equation is satisfied for a given $k$, we have:
$$
S\left(zc^\dagger-\overline{z}\,c\right)^{k+1}S^{-1}f=S\left(zc^\dagger-\overline{z}\,c\right)S^{-1}S\left(zc^\dagger-\overline{z}\,c\right)^{k}S^{-1}f
=S\left(zc^\dagger-\overline{z}\,c\right)S^{-1}\left(zb-\overline{z}\,a\right)^kf.
$$
Now, since $\left(zb-\overline{z}\,a\right)^kf\in\D$, it follows that
$$
S\left(zc^\dagger-\overline{z}\,c\right)S^{-1}\left(zb-\overline{z}\,a\right)^kf=\left(zb-\overline{z}\,a\right)\left(zb-\overline{z}\,a\right)^kf=\left(zb-\overline{z}\,a\right)^{k+1}f.
$$
Hence (\ref{35}) follows. Notice that all the equalities above are well defined since $\D$ is stable under the action of all the operators involved in our computation.

Now, let us compute $SW(z)S^{-1}f$. Because of the boundedness of $S$, $S^{-1}$ and $W(z)$, we have:
$$
SW(z)S^{-1}f=S\left(\sum_{k=0}^\infty\frac{1}{k!}\left(zc^\dagger-\overline{z}\,c\right)^k\right)S^{-1}f=$$
$$=\sum_{k=0}^\infty\frac{1}{k!}S\left(zc^\dagger-\overline{z}\,c\right)^kS^{-1}f=
\sum_{k=0}^\infty\frac{1}{k!}\left(zb-\overline{z}\,a\right)^kf.
$$
Then, since $SW(z)S^{-1}$ is bounded, the series $\sum_{k=0}^\infty\frac{1}{k!}\left(zb-\overline{z}\,a\right)^kf$ converges for all $z\in\Bbb C$ and for all $f\in\D$, and define $U(z)$ on $\D$.

\end{proof}

This Proposition implies that, if $S$ and $S^{-1}$ are both bounded, the three {\em displacement operators} $U(z)$, $V(z)$ and $W(z)$ are {\em almost similar}, meaning with this that a similarity map $S$ indeed exists, but the equalities in (\ref{34}) makes only sense, in general, on $\D$ and not on the whole $\Hil$. This can be understood easily: while $W(z)$, $S$ and $S^{-1}$ are bounded operators, $U(z)$ and $V(z)$ in general are unbounded, so they cannot be defined in all of $\Hil$.

\vspace{2mm}

An immediate and interesting consequence of the equations in (\ref{34}) is that $V(z)$ and $U(z)$ satisfy the following intertwining relation on $\D$:
\be
SS^\dagger V(z)f=U(z)SS^\dagger f
\label{36}\en
for all $f\in\D$. This may be relevant, since this kind of relations have several consequences in general. We refer to \cite{intop} for some results on intertwining operators. We will not insist on this aspect here, but still we want to stress that the operator doing the job, $SS^\dagger$, is close to $\Theta=S^\dagger S$, but $S$ and $S^\dagger$ appear in the reversed order. Of course, these two operators coincide if $S$ is self-adjoint.

\vspace{2mm}

Our results allow us to conclude (once more, see formula (\ref{32})) now that the two vectors in (\ref{32}) are well defined for all $z\in\Bbb C$, and, more interesting, that
\be
\varphi(z)=U(z)\varphi_0=S\Phi(z), \qquad \Psi(z)=V\Psi_0=(S^{-1})^\dagger\Phi(z),
\label{37}\en
for all $z\in\Bbb C$. The proof is straightforward and will not be given here. We just notice that, in particular, these equations imply that $\varphi_0\in D(U(z))$ and $\Psi_0\in D(V(z))$, $\forall\,z\in\Bbb C$.

In analogy with the notion of Riesz bases, formula (\ref{37}) suggests to introduce the notion of RBCS:

\vspace{2mm}

\begin{defn}\label{def3}
A pair of vectors $(\eta(z),\xi(z))$, $z\in\E$, for some $\E\subseteq\Bbb C$, are called RBCS if a standard coherent state $\Phi(z)$, $z\in\E$, and a bounded operator $T$ with bounded inverse $T^{-1}$ exists such that
\be
\eta(z)=T\Phi(z), \qquad \xi(z)=(T^{-1})^\dagger\Phi(z),
\label{38}\en

\end{defn}

It is clear then that $(\varphi(z),\Psi(z))$ are RBCS, with $\E=\Bbb C$. It is easy to check  that RBCS have a series of nice properties, which follow easily from similar properties of $\Phi(z)$. These properties are listed in the following proposition:

\begin{prop} Let $(\eta(z),\xi(z))$, $z\in\Bbb C$, be a pair of RBCS. Then:

(1) $$\left<\eta(z),\xi(z)\right>=1,$$
$\forall\,z\in\Bbb C$.

(2) For all $f,g\in\Hil$ the following equality ({\em resolution of the identity}) holds:
\be
\left<f,g\right>=\frac{1}{\pi}\int_{\Bbb C}d^2z \left<f,\eta(z)\right>\left<\xi(z),g\right>
\label{39}\en

(3) If a subset $\D\subset\Hil$ exists, dense in $\Hil$ and invariant under the action of $T^\sharp$, $(T^{-1})^\sharp$ and $c^\sharp$, and if the standard coherent state $\Phi(z)$ belongs to $\D$, then two operators $a$ and $b$ exist, satisfying (\ref{A1}), such that
\be
a\,\eta(z)=z\eta(z),\qquad b^\dagger\xi(z)=z\xi(z)
\label{310}\en
\end{prop}

\begin{proof}
The first statement is trivial and will not be proved here. As for the second, due to the fact that both $T$ and $T^{-1}$ in Definition \ref{def3} are bounded, we have, for all $f,g\in\Hil$,
$$
\left<f,g\right>=\left<T^\dagger f,T^{-1}g\right>=\frac{1}{\pi}\int_{\Bbb C}d^2z \left<T^\dagger f,\Phi(z)\right>\left<\Phi(z),T^{-1}g\right>=
$$
$$
=\frac{1}{\pi}\int_{\Bbb C}d^2z \left< f,T\Phi(z)\right>\left<(T^{-1})^\dagger\Phi(z),g\right>=\frac{1}{\pi}\int_{\Bbb C}d^2z \left<f,\eta(z)\right>\left<\xi(z),g\right>,
$$
because of (\ref{38}).
To prove (3) we first observe that  our assumption implies that the two operators $a$ and $b$ defined as  $a=TcT^{-1}$ and $b=Tc^\dagger T^{-1}$ map $\D$ into $\D$, and that $[a,b]f=f$ for all $f\in\D$. The eigenvalue equations in (\ref{310}) simply follow now from (\ref{38}).

\end{proof}

It is interesting to notice that the resolution of the identity is valid in all of $\Hil$. This is true in the present settings, but we do not expect a similar result can be established if Assumption $\D$-pbs 3 is replaced with one of its weaker versions. We refer to \cite{abg2015} for some results concerning this situation. Concerning the saturation of the Heisenberg uncertainty relation, this cannot be recovered by these RBCS using the standard, self-adjoint, position and momentum operators $q$ and $p$. However, if $q=\frac{1}{\sqrt{2}}(c+c^\dagger)$  and $p=i\,\frac{1}{\sqrt{2}}(c^\dagger-c)$ are replaced by  $Q=\frac{1}{\sqrt{2}}(a+b)$ and $P=i\,\frac{1}{\sqrt{2}}(b-a)$, then we believe that a {\em deformed} version of the Heisenberg uncertainty relation involving these operators can, in fact, be saturated. This aspect will be discussed in a future paper, together with several examples of RBCS. Here we just consider a first simple example of these states, related to the harmonic oscillator.

\vspace{2mm}

{\bf An example from the harmonic oscillator:--} Let $\Phi(z)$ be the standard coherent state arising in the treatment of the quantum harmonic oscillator with Hamiltonian $H=c^\dagger c+\frac{1}{2}\,\1$, $[c,c^\dagger]=\1$. In the coordinate representation this state, which we indicate here $\Phi_z(x)$, $z\in \Bbb C$ and $x\in \Bbb R$, is the solution of $c\,\Phi_z(x)=z\Phi_z(x)$. With a suitable choice of normalization we have
$$
\Phi_z(x)=\frac{1}{\pi^{1/4}}\,e^{-\frac{1}{2}x^2+\sqrt{2}zx-\Re(z)^2}.
$$
Now, let $P=|e_0\left>\right<e_0|$ be the orthogonal projector operator on the ground state $e_0(x)=\frac{1}{\pi^{1/4}}\,e^{-\frac{1}{2}x^2}$ of the harmonic oscillator. Then the operator $T=\1+iP$ is bounded, invertible, and its inverse, $T^{-1}=\1-\frac{1+i}{2}\,P$, is also bounded. Hence we can use formula (\ref{38}) deducing that
$$
\varphi_z(x)=T\Phi_z(x)=e_0(x)\left(e^{\sqrt{2}zx-\Re(z)^2}+ie^{-\frac{1}{2}|z|^2+\frac{i}{2}\Re(z)\Im(z)}\right),
$$
while
$$
\Psi_z(x)=(T^{-1})^\dagger\Phi_z(x)=e_0(x)\left(e^{\sqrt{2}zx-\Re(z)^2}-\frac{1-i}{2}\,e^{-\frac{1}{2}|z|^2+\frac{i}{2}\Re(z)\Im(z)}\right).
$$
These are our RBCS, in coordinate representation. They both appear to be suitable deformations of the original vector $\Phi_z(x)$. It is not hard to imagine how to generalize this construction: it is enough to replace the operator $P$ with some different orthogonal projector, for instance with the projector on a given normalized vector $u(x)$, $P_u=|u\left>\right<u|$, $u(x)\neq e_0(x)$.

\section{Conclusions}

We have seen how bounded operators with bounded inverse can be used to construct not only Riesz biorthogonal bases, but also bicoherent states, having several properties which are  similar to those of standard coherent states. More important, we have seen that these RBCS are naturally related to $\D$-PBs of a particular kind, the ones for which Assumption $\D$-pbs 3 holds true. It is clear that what we have discussed here is just the beginning of the story. There are several aspects of RBCS which deserve a deeper analysis. Among them, we cite the (maybe) most difficult: what does it happen if Assumption $\D$-pbs 3 is not satisfied? And, more explicitly, what can be said when Assumption $\D$-pbw 3 is true? This is much harder, but possibly more interesting in concrete physical applications, since in this case, even if we can introduce a pair of bicoherent states, \cite{abg2015}, in general there is no bounded operator with bounded inverse mapping these states into a single standard coherent state. Moreover, we have several problems with the domain of the unbounded operators appearing in the game, and this, of course, requires more (and more delicate) mathematics.

Another aspect, which was just touched in \cite{abg2015}, but not here, and which surely deserves a deeper analysis, is the use of bicoherent states, of the Riesz type or not, in quantization procedures. This may be relevant in connection with non conservative systems, or with physical system described by non self-adjoint Hamiltonians.

Another interesting open problem, which has been widely considered for standard coherent states along the years, is to check if completeness can be recovered for some suitable discrete subset of RBCS, i.e. if we can fix a discrete lattice in $\Bbb C$, $\Lambda:=\{z_j\in\Bbb{C},\,j\in\Bbb N\}$, such that the set $\left\{(\eta(z_j),\xi(z_j)), \,z_j\in\Lambda\right\}$ is reach enough to produce a resolution of the identity in $\Hil$. Stated in a different way, is it possible to extend the results deduced in \cite{zak} for standard coherent states to RBCS or to bicoherent states in general? We believe that this can in fact be done for RBCS, while for general bicoherent states this is not so evident.

\section*{Acknowledgements}

The author acknowledges partial support from Palermo University and from G.N.F.M. of the INdAM.

\end{document}